\documentclass[letterpaper,conference,10pt]{IEEEtran}
\usepackage{amsmath,epsfig,cite}
\usepackage{amssymb, amsthm}
\usepackage{algorithm,algorithmic}

\newtheorem{theorem}{Theorem}

\begin{document}

\title{Generation of Innovative and Sparse Encoding
Vectors for Broadcast Systems with Feedback
\thanks{This work was partially supported by a grant from the University Grants Committee of the Hong Kong Special Administrative Region, China (Project No. AoE/E-02/08).}}

\author{
\IEEEauthorblockN{Ho Yuet Kwan}
\IEEEauthorblockA{Department of Electronic Engineering\\
City University of Hong Kong\\
Email: hykwan@cityu.edu.hk} \and
\IEEEauthorblockN{Kenneth W. Shum}
\IEEEauthorblockA{Institute of Network Coding\\
The Chinese University of Hong Kong\\
Email: wkshum@inc.cuhk.edu.hk}
\and
\IEEEauthorblockN{Chi Wan Sung}
\IEEEauthorblockA{Department of Electronic Engineering\\
City University of Hong Kong\\
Email: albert.sung@cityu.edu.hk}
}

\IEEEoverridecommandlockouts

\maketitle

\begin{abstract}
In the application of linear network coding to wireless broadcasting with feedback, we prove that the problem of determining the existence of an innovative encoding vector is NP-complete when the finite field size is two. When the finite field size is larger than or equal to the number of users, it is shown that we can always find an encoding vector which is both innovative and sparse. The sparsity can be utilized in speeding up the decoding process.  An efficient algorithm to generate innovative and sparse encoding vectors is developed. Simulations show that the delay performance of our scheme with binary finite field outperforms a number of existing schemes in terms of average and worst-case delay.
\end{abstract}

\section{Introduction}
Linear network coding provides an excellent solution to the wireless broadcasting problem in terms of reliability and channel utilization \cite{LNC1,LNC2}. The idea is to send encoded packets that are obtained by taking linear combinations over a finite field of all the original packets. The encoding vector specifies the coefficients for the linear combination and is determined by the transmitter. An encoded packet together with a header which contains the corresponding encoding vector is broadcasted to all users. An encoded packet is said to be {\em innovative to a user} if the corresponding encoding vector is not in the subspace spanned by the encoding vectors already received by that user. It is called {\em innovative} if it is innovative to all users who have not yet received enough packets for decoding. Obviously, if all the encoded packets generated by the transmitter for transmission are {\em innovative}, the total number of packet transmissions for all users to obtain the complete set of packets can be minimized.

To utilize the radio channel efficiently, it is important to generate innovative packets. In~\cite{DFT07}, it is shown that if the size of the finite field is equal to the number of users, an innovative packet can always be found. In \cite{KDF08}, the authors consider a system with perfect feedback, in which the transmitter knows the status of all users and tries to find an innovative packet by a probabilistic algorithm. This approach is shown to be {\em rate-optimal} if the underlying finite field is sufficiently large. The average number of transmissions is analyzed in~\cite{XYWZ08}. By exploiting the feedback information from users, some authors develop algorithms to generate instantly decodable network-coded packets in \cite{ID1, ID2} so that innovative packets can be decoded once the packets are available at the receivers without waiting for the complete reception of the full set of packets. Randomized broadcast coding without utilizing any feedback is analyzed in~\cite{EOM06}. The computation with large finite field may be costly for mobile hand-held devices. In order to reduce encoding and decoding complexity, random linear network code over the binary field without feeding back what the receivers have received is considered in~\cite{GTK08,HPFL09}. This approach lowers computational complexity at the expense of larger number of retransmissions.

When the finite field size is small, an innovative encoding vector may not exist. In Section~\ref{sec:NPhard}, we prove that the problem of determining the existence of innovative encoding vector is NP-complete. A related result is also obtained in~\cite{ERCS07}, where the broadcast channel is assumed to be noiseless, and the problem of minimizing the number of packets required to finish off the file transmission with the binary finite field is shown to be NP-complete. In Section IV, we show that we can always find a sparse and innovative encoding vector with at most $K$ non-zero components using a deterministic algorithm called the cofactor method. In Section~\ref{sec:PE}, the cofactor method is compared with some other transmission schemes by simulation, for both small and large finite fields.


\section{System Model and Problem Formulation}
We consider a wireless single-hop system consisting of one transmitter and $K$ receivers/users. We denote $S$ as the source node and $U_i$ as the $i$-th receiver, where $i \in \{1,2,\ldots,K\}$.
The source node $S$ wants to broadcast a file to all receivers via a wireless channel, which is modeled as a broadcast erasure channel.
The input to the broadcast erasure channel is a $q$-ary alphabet. We will also call a $q$-ary alphabet a packet.
Each receiver successfully receives the transmitted packet with probability $1-P_e$, independent of each other, where $P_e$ denotes the erasure probability. An erased packet is unrecoverable and discarded, while a successfully received packet is assumed to be error-free. We assume that there is a feedback channel from each receiver to the source. Upon receiving a packet successfully, a receiver sends an acknowledgement to the source node. We assume that the feedback channel has no delay and no error. The source node keeps track of the status of each receiver. The transmitted packet is a function of the source file and the acknowledgements from the $K$ receivers.

In this paper, we focus on transmission schemes with linear network coding. The alphabet size $q$ is a power of prime and the alphabet set is identified with the finite field $GF(q)$.
The file is packetized into $N$ packets. The transmitted packet is a linear combination of the $N$ packets, with coefficients drawn from $GF(q)$. An encoding vector is an $N$-vector whose components are the $N$ coefficients used in the generation of a transmitted packet. Each user returns an acknowledgement to the source node if a packet is received successfully, until he has already received $N$ packets whose encoding vectors are linearly independent over $GF(q)$. In order to minimize the delay of each user, it is crucial to generate encoding vectors which are innovative to all users.

Our objectives are  (i) to determine whether an innovative encoding vector exists, and (ii) to devise an effective algorithm for generating innovative and sparse encoding vectors.

\section{NP-Completeness in Finding an Innovative Encoding Vector when $q=2$}
\label{sec:NPhard}

If the field size $q$ is larger than or equal to $K$, it is known that an innovative encoding can always be found~\cite{DFT07}. Indeed, the number of non-zero encoding vectors which are {\em not} innovative to user $i$ is equal to $q^{r_i}$, where $r_i$ is the rank of the subspace spanned by the encoding vectors already received by user~$i$ and $r_i < N$ , and hence by the union bound, the number of non-zero vectors which are not innovative is at most
$  (q^{r_1}-1)+(q^{r_2}-1) + \ldots + (q^{r_K}-1)$.
When $q \geq K$, the number of non-zero and non-innovative encoding vector is strictly less than $q^N-1$, and hence we can always find an innovative packet.

When the underlying finite field is small, an innovative encoding vector may not exist.

{\bf Problem} $q$-$\mathsf{IEV}$: A problem instance consists of $K$ matrices $\mathbf{C}_i$ over $GF(q)$, $i=1,2,\ldots, K$, and each matrix has $N$ columns. Determine whether there is an $N$-dimensional vector over $GF(q)$ which is not in the row space of $\mathbf{C}_i$, for all $i$.

\begin{theorem}
2-$\mathsf{IEV}$ is NP-complete.
\end{theorem}

\begin{proof}
The idea is to reduce the 3-$\mathsf{SAT}$ problem, well-known to be NP-complete~\cite{GareyJohnson}, to the 2-$\mathsf{IEV}$ problem. Recall that the 3-$\mathsf{SAT}$ problem is a Boolean satisfiability problem, whose instance is a Boolean expression written in conjunctive normal form with three variables per clause (3-$\mathsf{CNF}$), and the question is to decide if there is some assignment of TRUE and FALSE to the variables such that the given Boolean expression has a TRUE value.

Let $E$ be a given Boolean expression with $n$ variables $x_1,\ldots, x_n$, and $m$ clauses in 3-CNF. We want to reduce the 3-$\mathsf{SAT}$ problem to the 2-$\mathsf{IEV}$ problem with $N=n+1$ packets and $K=m+1$ users.

To the $i$-th clause ($i=1, 2, \ldots, m$), we first construct a $3\times N$ matrix $\mathbf{B}_i$. If the $j$-th literal ($j=1,2,3$) in the $i$-th clause is $x_k$, then let the $k$-th component in the $j$-th row of $\mathbf{B}_i$ be 1, and the other elements be all zero. Otherwise, if the $j$-th literal in the $i$-th clause is $\neg x_k$, then let the $k$-th and the $(n+1)$-st component in the $j$-th row of $\mathbf{B}_i$ be 1, and the remaining components be all zero. Let $\mathbf{C}_i$ be the matrix whose rows form a basis of the orthogonal complement of the row space of $\mathbf{B}_i$. We will use the fact that a vector $\mathbf{v}$ is in the row space of $\mathbf{C}_i$ if and only if $\mathbf{B}_i \mathbf{v}^T = 0$.

Consider an example with $n=4$ Boolean variables. From the clause $\neg x_1 \vee \neg x_2 \vee  x_3$, we get
\[\small \mathbf{B}_i =
\begin{bmatrix}
1& 0& 0& 0& 1 \\
0& 1& 0& 0& 1 \\
0& 0& 1& 0& 0
\end{bmatrix}, \
 \mathbf{C}_i = \begin{bmatrix}
 0 & 0 & 0 & 1 & 0 \\
 1 & 1 & 0 & 0 & 1
 \end{bmatrix}.
\]
It can be verified that each row in $\mathbf{B}_i$ is orthogonal to the rows in $\mathbf{C}_i$, i.e.,
the row space of $\mathbf{C}_i$ is the orthogonal complement of the row space of $\mathbf{B}_i$.

For the extra user, user $m+1$, let $\mathbf{B}_{m+1}$ be the $1\times (n+1)$ matrix $[\mathbf{0}_n 1]$, where $\mathbf{0}_n$ stands for the $1\times n$  all-zero vector. The problem reduction can be done in polynomial time.

Let $\mathbf{x} = [x_1\ x_2\ \ldots x_n]$ be a Boolean row vector and $\hat{\mathbf{x}} = [\mathbf{x}\;1]$. Obviously, any solution $\mathbf{x}$ to a 3-$\mathsf{SAT}$ problem would cause the product $\mathbf{B}_j \hat {\mathbf{x}}^T$ a non-zero vector for $j=1, 2, \ldots, m$ and $[\mathbf{0}_n 1]\hat {\mathbf{x}}^T \neq 0$. Therefore $\hat{\mathbf{x}}$ is not in the row space of $\mathbf{C}_j$ for all $j$. Hence $\hat{\mathbf{x}}$ is also a solution to the derived 2-$\mathsf{IEV}$ problem.


Conversely, any solution to the derived 2-$\mathsf{IEV}$ problem also yields a solution to the original 3-$\mathsf{SAT}$  problem as well. Let $\mathbf{c} = [c_1\ c_2\ \ldots c_n\ c_{n+1}] \in GF(2)^{n+1}$ be a solution to the derived 2-$\mathsf{IEV}$ problem. Note that we must have $c_{n+1} = 1$ because of $\mathbf{B}_{m+1}$. Let  $i$ be an integer between 1 and~$m$. Since $\mathbf{c}$ is not in the row space of $\mathbf{C}_i$, the product $\mathbf{B}_i \mathbf{c}^T$ is a non-zero vector, for otherwise $\mathbf{c}$ would belong to the orthogonal complement of $\mathbf{B}_i$. Hence, if we assign TRUE to $x_k$ if $c_k=1$ and FALSE to $x_k$ if $c_k=0$, for $k=1,2,\ldots,n$, then the $i$-th clause will have a TRUE value. Since this is true for all $i$, the whole Boolean expression also has a TRUE value.

The problem 2-$\mathsf{IEV}$ is clearly in NP, since it is efficiently verifiable. Hence it is NP-complete.
\end{proof}

\section{Generation of Sparse Encoding Vectors} \label{sec:SEV}

After receiving $N$ packets whose encoding vectors are linearly independent over $GF(q)$, a user can recover the source data by solving a system of $N$ linear equations. The standard Gaussian elimination requires $O(N^3)$ operations over $GF(q)$. One way to reduce the decoding complexity is to choose encoding vectors which are sparse. A vector is called {\em $w$-sparse} if there are no more than $w$ non-zero components.

\begin{theorem}
If $q \geq K$, we can find an innovative encoding vector which is $K$-sparse. \label{thm:cofactor}
\end{theorem}

\begin{proof}
Suppose that user $k$, for $k=1,2,\ldots, K$, has received $r_k$ packets whose encoding vectors are linearly independent.
Let $\mathbf{C}_k$ be the $r_k \times N$ matrix obtained by putting together the $r_k$ encoding vectors. We want to find a $K$-sparse innovative encoding vector $\mathbf{x} = [x_1\ x_2\ \ldots \ x_N]$.

Since $\mathbf{C}_k$ is full-rank, we can find $r_k$ columns of $\mathbf{C}_k$ which are linearly independent. Let $\mathcal{I}_k$ be a set of indices of $r_k$ linear independent columns in $\mathbf{C}_k$. For each $k$, we arbitrarily pick a column whose index is not in $\mathcal{I}_k$.  We call this the {\em extra column} and let $\mathcal{I}_k'$ be the union of $\mathcal{I}_k$ and the index of this extra column. The cardinality of $\mathcal{I}_k'$ is $r_k+1$. For each $k=1,2,\ldots, K$, we construct an $(r_k+1) \times (r_k+1)$ matrix $\hat{\mathbf{H}}_k$, by first appending the vector $\mathbf{x} = [x_1 \ x_2 \ \ldots \ x_N]$ to the bottom of matrix $\mathbf{C}_k$, and then deleting all columns of the resulting matrix except the columns with indices in $\mathcal{I}_k'$.

For each $k$, we compute the $r_k+1$ cofactors of the entries in the last row of $\hat{\mathbf{H}}_k$. Let $x_{i_k}$ be the variable with largest index in the last row of $\hat{\mathbf{H}}_k$ whose cofactor is non-zero. The column indices $i_1,\ldots, i_K$ so obtained may not be distinct. Let $\mathcal{J} \triangleq \{j_1,j_2,\ldots, j_s\}$ be the set of distinct indices such that $\mathcal{J} = \{i_1,i_2,\ldots, i_K\}$ and $j_1 < j_2 < \ldots< j_s$. Also, for $t=1,2,\ldots, s$, we let $\mathcal{K}_t$ be the set of users such that $k\in\mathcal{K}_t$ if and only if $i_k = j_t$.
We remark that $\mathcal{J}$ contains at most $K$ distinct indices, i.e., $s\leq K$.

We obtain a $K$-sparse innovative encoding vector as follows. First, we set all variables $x_i$, for $i\notin \mathcal{J}$, to zero. Then we assign values to $x_{j_1}, x_{j_2}, \ldots, x_{j_s}$ sequentially, so that the determinant of $\hat{\mathbf{H}}_k$ is nonzero for all $k$. For $k \in \mathcal{K}_1$, the last row of $\hat{\mathbf{H}}_k$ has only one variable, namely $x_{j_1}$, whose value is not yet assigned (the rest are all set to zero). The cofactor of $x_{j_1}$ in $\hat{\mathbf{H}}_k$ is non-zero. If we expand the determinant of $\hat{\mathbf{H}}_k$ in the last row, we see that the determinant can be written as $b_{i_k} x_{i_k}$, where $b_{i_k}$ is the cofactor of $x_{i_k}$ in  $\hat{\mathbf{H}}_k$. We have a non-zero value if $x_{j_1}$ is non-zero, for all  $k \in \mathcal{K}_1$. We can assign any non-zero element of $GF(q)$ to $x_{j_1}$, and make $\hat{\mathbf{H}}_k$ non-zero for all $k\in\mathcal{K}_1$.

Inductively, suppose that the values of $x_{j_1}, \ldots x_{j_{t-1}}$ have been assigned. Consider the determinants of $\hat{\mathbf{H}}_k$ for $k\in \mathcal{K}_t$. The only variable in the last row of $\hat{\mathbf{H}}_k$ which has not been assigned a value yet is $x_{j_t}$. If we expand the determinant on the last row, we obtain a linear polynomial in the form of $a_{i_k} + b_{i_k}x_{j_t}$, where $a_{i_k}$ is a constant and $b_{i_k}$ is the cofactor of $x_{j_t}$ in $\hat{\mathbf{H}}_k$. There are at most $K$ such degree-one polynomials, and thus we can assign a value to $x_{j_t}$ such that all determinants of $\hat{\mathbf{H}}_k$ are non-zero. Here we have used the assumption that $q \geq K$. Note that the assignment of $x_{j_t}$ does not affect the determinants of previous users with indices in $\mathcal{K}_1  \cup \mathcal{K}_2 \cup \cdots \cup \mathcal{K}_{t-1}$, because $j_t$ either does not appear in $\mathcal{K}_1  \cup \cdots \cup \mathcal{K}_{t-1}$ or the corresponding cofactor in $\hat{\mathbf{H}}_\ell$ is equal to zero for $\ell\in \mathcal{K}_1  \cup \cdots \cup \mathcal{K}_{t-1}$.

After the end of the process, we have chosen the values for $x_1,x_2,\ldots, x_N$ such that $|\hat{\mathbf{H}}_k|$ is non-zero for all $k=1,2,\ldots, K$. This encoding vector is innovative to all users and contains at most $K$ non-zero components.
\end{proof}

We call the the method described in the proof of Theorem~\ref{thm:cofactor} the {\em cofactor method}. Using the cofactor method, we can produce innovative and $K$-sparse encoding vectors. For the decoding, the number of non-zero coefficients in the linear system is no more than $KN$.

{\em Example 1.}
When $q=3$, $K=3$ and $N=3$, consider

\vspace{-0.5cm}

\[ \small
 \mathbf{C}_1 = \begin{bmatrix}
 0 & 1 &0
 \end{bmatrix}, \
 \mathbf{C}_2 = \begin{bmatrix}
 1 & 0 & 1 \\
 0 & 1 & 1
 \end{bmatrix},
\ \mathbf{C}_3 = \begin{bmatrix}
1 & 0 & 0 \\
0 & 2 & 0
 \end{bmatrix}.
\]
Suppose that $\mathcal{I}_1' = \{1, 2\}$, $\mathcal{I}_2'=\mathcal{I}_3'=\{1,2,3\}$. We have
\[ {\small
  \mathbf{\hat H}_1\! =\! \begin{bmatrix}
 0 & 1  \\
 x_1 & x_2
 \end{bmatrix},
 \mathbf{\hat H}_2\! =\! \begin{bmatrix}
 1 & 0 & 1 \\
 0 & 1 & 1 \\
 x_1 & x_2 & x_3  \end{bmatrix},
\mathbf{\hat H}_3\! =\! \begin{bmatrix}
1 & 0 & 0 \\
0 & 2 & 0\\
 x_1 & x_2 & x_3 \end{bmatrix}}.
\]
In $\hat{\mathbf{H}}_1$ the cofactors of $x_1$ and $x_2$ are $-1$ and 0 respectively. Hence, $i_1 = 1$.  In $\hat{\mathbf{H}}_2$, all cofactors of $x_1$, $x_2$ and $x_3$ are non-zero. We thus have $i_2 = 3$. In $\hat{\mathbf{H}}_3$, the cofactor of $x_3$ is nonzero, and so $i_3=3$. The index set $\mathcal{J}$ is equal to $\{1,3\}$. The two index sets of users are $\mathcal{K}_1 = \{1\}$ and $\mathcal{K}_2 =\{2,3\}$.
By the cofactor method, we first assign  0 to $x_2$. Then we go through the variable indices in $\mathcal{J}$ in ascending order. For $x_1$, we can assign any nonzero value to $x_1$. For example, we pick $x_1 = 1$. Once $x_1$ and $x_2$ are fixed, we compute the determinants of $\hat{\mathbf{H}}_2$ and $\hat{\mathbf{H}}_3$, which are $-1+x_3$
and $2x_3$ respectively.
Finally, we want to assign a value to  $x_3$ such that $-1+x_3\neq 0$ and $2x_3\neq 0$. The only choice in this example is  $x_3=2$. The resulting encoding vector is $[1 \ 0 \ 2]$.

\smallskip

In the cofactor method, the main complexity is related to the computation of $r+1$ cofactors in an $(r+1) \times (r+1)$ matrix. A straightforward calculation of an $r\times r$ determinant requires $O(r^3)$ arithmetic operations. The calculation of all cofactors in a matrix would require $O(r^4)$ operations per each user in each step. We can use a more efficient algorithm, called the {\em Bareiss algorithm}. The number of arithmetic operations over $GF(q)$ required in the computation of cofactors per user can be reduced to~$O(N^3)$. Summing over all $K$ users, the complexity for computing all cofactors is $O(KN^3)$. The complexity of the rest of the cofactor method is of $O(K^2N)$. The overall complexity of the cofactor method is $O(KN^3+K^2N)$. If Jaggi-Sanders algorithm in \cite{JSCEEJT} is applied to solve the encoding problem, the complexity is $O(K N^2 (K+N))$. It means that the cofactor method is no worse than the algorithm in \cite{JSCEEJT} in terms of the encoding complexity. But certainly the encoding vector produced by Jaggi-Sanders algorithm is not sparse. Details on the Bareiss algorithm is given in the appendix.

The cofactor method assumes that $q\geq K$. If $q < K$, the cofactor method may fail to find an assignment of the $x_i$'s such that all determinants are non-zero. In that case, we set those $x_i$'s to zero and the encoding vectors so generated may not be innovative. But anyway, the returned encoding vector is  $K$-sparse. Hence we can still apply the cofactor method for the case $q=2$ to obtain $K$-sparse encoding vectors, which are innovative to only a fraction of the $K$ users.


In \cite{SSK11}, the problem of generating the sparest innovative is considered, and is shown to be NP-hard when $q\geq K$.

\section{Performance Evaluation} \label{sec:PE}

We evaluate the cofactor method via simulations. In the simulations, we divide the transmission into two phases. The source node first transmits all packets one by one uncoded. The $K$ users acknowledge the packets they have successfully received. The source node sets up $K$ matrices $\mathbf{C}_k$, for $k=1,2\ldots, K$. The rows of $\mathbf{C}_i$ are the encoding vectors received by user~$k$. Since the packets are uncoded in the first phase, each row of $\mathbf{C}_k$ contains exactly one nonzero component. We initialize $\mathcal{I}_i$ to be the set of non-zero columns in $\mathbf{C}_k$.  In the second phase, we transmit the packets using the encoding vectors generated by the cofactor method.

Each simulation points involved 1000 random realizations and we assume that $N=32$ and $P_e=0.3$. The worst-case delay is defined as the average of total number of transmissions for $S$ to ensure that all users receive an intact file over 1000 random realizations. The average delay means the average number of transmissions for $S$ so that an intact file can be received by a user. For the decoding complexity, we count the number of additions and multiplications in decoding. In our simulations, an addition operation involving two non-zero operands is counted. A multiplication operation is counted when none of the two operands is 1 or 0.

\begin{figure}
\epsfxsize=2.7in \centerline{\epsffile{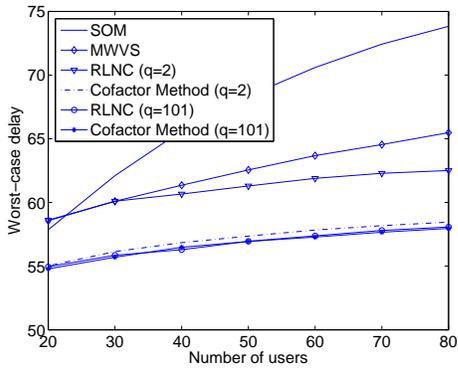}} \caption{The worst-case delay vs the number of users} \label{f1}
\end{figure}

\begin{figure}
\epsfxsize=2.7in \centerline{\epsffile{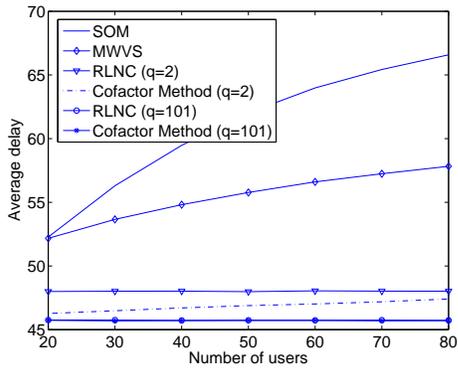}} \caption{The average delay vs the number of users} \label{f2}
\end{figure}

\begin{figure}
\epsfxsize=2.7in \centerline{\epsffile{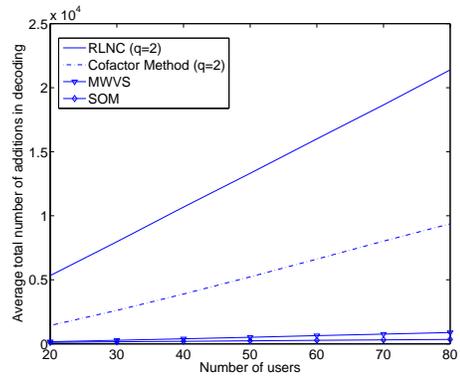}} \caption{The average total no. of additions vs the number of users ($q$=2)} \label{f3}
\end{figure}

\begin{figure}
\epsfxsize=2.7in \centerline{\epsffile{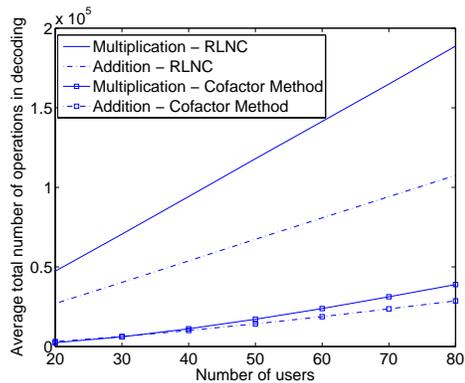}} \caption{The average total no. of operations vs the number of users  ($q$=101)} \label{f4}
\end{figure}

Figure~\ref{f1} shows the worst-case delay performance of our system with the cofactor method, the random linear network code (RLNC) scheme in which the components are selected according to a uniform distribution, the sorted opportunistic method (SOM) in \cite{ID1} and the maximum weight vertex search (MWVS) algorithm in \cite{ID2} for encoding vector generations, where both SOM and MWVS generate instantly decodable packets. It is found that, for $q=2$, the cofactor method always performs better than RLNC, SOM and MWVS in terms of the worst-case delay. In addition, we also find that the worst-case delay performance of the cofactor method with a small finite field size ($q=2$) is comparable to that of RLNC with a large finite field size ($q=101$). Next, we consider the average delay performance. According to information theory, the best we can do is to have $N/(1-P_e) = 45.7$ transmissions on average. In Figure~\ref{f2}, we observe that both the cofactor method and RLNC with large enough finite field size ($q=101$) can achieve the limit. From the figure, we also see that although all concerned methods may not be optimal when $q=2$, the cofactor method always results in a smaller average delay.

The decoding algorithms in most of the previous work are basically Gauss-Jordan elimination except the instantly decodable schemes in \cite{ID1, ID2}. We implement the Gauss-Jordan elimination for sparse matrix in our simulation. Note that the $K$-sparse property of the cofactor method implies an upper bound on the number of non-zero entries in an encoding vector. In practice, the average number of non-zero entries is significantly less than both $K$ and $N$ even for $K > N$. As a result, significant decoding complexity reduction is expected for a system with the cofactor method. Figures~\ref{f3} and~\ref{f4} show the average total number of operations for all users in the system when $q=2$ and $q=101$, respectively. The cofactor method indeed yields significant reduction in both the average total number of addition and multiplication operations when compared with RLNC. From Figure~\ref{f3}, we observe that, with both SOM and MWVS which are instantly decodable, a receiver enjoys a low decoding complexity at the expense of larger delay. As a result, the cofactor method which always generates sparse encoding vectors is a promising choice in terms of delay performance and decoding complexity.

\section{Conclusions}
We devise a cofactor method to generate $K$-sparse encoding vector. When $q\geq K$, it is guaranteed that the resulting encoding vector is innovative, and hence the broadcast system is delay-optimal.
The sparsity can be exploited in devising faster decoding algorithm. Simulation result shows that the cofactor method outperforms RLNC, SOM and MWVS in terms of both the worst-case delay and average delay. On the other hand, when $q=2$, the problem of determining the existence of an innovative encoding vector is NP-complete.


\appendix

\label{app}
The Bareiss algorithm is a fraction-free algorithm for computing determinant~\cite{Yap2000}. To illustrate the idea, we apply Bareiss algorithm to an $n \times n$ matrix $\mathbf{M}$ whose elements are integers and the last row consists of indeterminates $x_1$, $x_2, \ldots, x_n$. At the end of the algorithm, the entry in the lower-right corner of $\mathbf{M}$ is a linear polynomial in $x_1$, $x_2,\ldots, x_n$, and the coefficient of $x_i$ is the corresponding cofactor of $x_i$ in the original $\mathbf{M}$. We remark that this is an in-place algorithm, and the complexity is in the order of $n^3$.

\begin{algorithm}
   \caption{Bareiss algorithm}
   \label{algorithm:Bareiss}
   {\em Input:} An $n \times n $ matrix $\mathbf{M}$. Assume that all principle minors of $\mathbf{M}$ are nonzero.\\
   {\em Notation:} Let $m_{ij}$ denote the $(i,j)$-entry of $\mathbf{M}$, and $m_{00} \triangleq 1$.

   \begin{algorithmic}
       \FOR{$k = 1,\ldots, n-1$}
           \STATE Compute $m_{ij} \leftarrow \frac{m_{kk}m_{ij} - m_{ik} m_{kj}}{m_{k-1,k-1}}$,
           for $i,j = k+1, \ldots, n$.
       \ENDFOR
   \end{algorithmic}
   {\em Output:} Return the $(n,n)$-entry of $\mathbf{M}$.
\end{algorithm}

{\em Example 2.} Consider
${\small
 \mathbf{M} = \begin{bmatrix} 1 & 2 & 3 \\ 4 & 5 & 6 \\ x_1 &x_2 & x_3 \end{bmatrix}}$,
as an example. After the first pass of the for-loop ($k=1$), the partial result is
\[ {\small
 \mathbf{M} = \begin{bmatrix}
 1 & 2 & 3 \\
 4 & -3 & -6 \\
 x_1 &x_2-2x_1 & x_3-3x_1 \end{bmatrix}}.
\]
After the end of the algorithm, we have
\[ {\small
 \mathbf{M} = \begin{bmatrix}
 1 & 2 & 3 \\
 4 & -3 & -6 \\
 x_1 &x_2-2x_1 & -3x_1 + 6x_2-3x_3 \end{bmatrix}}.
\]
The coefficients of $x_1$, $x_2$ and $x_3$ of the polynomial in the $(3,3)$-entry are the cofactors
$
\begin{vmatrix} 2&3 \\ 5&6 \end{vmatrix}$,
$-\begin{vmatrix} 1&3 \\ 4&6 \end{vmatrix}$, and
$\begin{vmatrix} 1&2 \\ 4&5 \end{vmatrix}$ respectively.


Furthermore, the algorithm can be run incrementally. Suppose that only the first $r$ rows in a matrix $\mathbf{C}$ is available,
\vspace{-0.05cm}
 \[ {\small \mathbf{C} =
\begin{bmatrix}
 c_{11} & c_{12} & \cdots & c_{1r} & \cdots & c_{1n}\\
\vdots & \vdots & \ddots & \vdots & \vdots &\vdots \\
c_{r,1} & c_{r,2} & \cdots & c_{r,r} & \cdots & c_{r,n}\\
* & * & \cdots & * &* &*\\
x_1 & x_2 & \cdots & x_r & \cdots & x_{n}
\end{bmatrix}}.
\]
The entries marked by ``$*$'' are not known yet and will be revealed later. We can apply the Bareiss algorithm to the submatrix obtained by removing the ``$*$'' entries and the right $n-r$ columns. When the value of the $(r+1)$-st row is known, we can run the Bareiss algorithm again on the submatrix obtained by removing rows $r+2$ to $n-1$ and the $n-r-1$ columns on the right. We can see that the $r^2$ entries in the first $r$ rows and the first $r$ columns are the same as before and we do not need to re-calculate them. Only the calculation of the $2r+1$ new entries are required. For each user, the source node essentially runs the Bareiss algorithm on an $N\times N$ matrix, and the complexity per user is $O(N^3)$. Summing over all users, the complexity involving the computation of the cofactors is $O(K N^3)$.

\smallskip
Acknowledgement: The authors are grateful to Prof. Wai Ho Mow and Dr. Kin-Kwong Leung for fruitful email discussions.




\end{document}